\date{}
\documentclass[11pt]{article}
\usepackage{hyperref}
\usepackage{tikz}
\usetikzlibrary{snakes}
\usetikzlibrary{decorations.text,calc,arrows.meta}
\usepackage[vcentermath]{youngtab}
\usepackage{tkz-euclide}
\usepackage{pgfplots}
\usetikzlibrary{decorations.shapes}
\usetikzlibrary{positioning}
\usetikzlibrary{patterns}
\usetikzlibrary{shapes.geometric}
\usetikzlibrary{decorations.pathmorphing}
\usetikzlibrary{arrows.meta}
\tikzset{snake it/.style={decorate, decoration=snake}}
\usetikzlibrary{arrows,shapes,positioning}
\usetikzlibrary{decorations.markings}
\tikzstyle arrowstyle=[scale=1]
\tikzstyle directed=[postaction={decorate,decoration={markings,mark=at position .65 with {\arrow[arrowstyle]{stealth}}}}]
\tikzstyle reverse directed=[postaction={decorate,decoration={markings,mark=at position .65 with {\arrowreversed[arrowstyle]{stealth};}}}]

\usepackage{amsmath,amsthm,amsfonts,amssymb,amsopn,amscd} 
\usepackage{color}

\def\om{\omega}
\def\Om{\Omega}

\newcommand{\bthm}{\begin{theorem}}
\newcommand{\ethm}{\end{theorem}}

\newcommand{\bprop}{\begin{proposition}}
\newcommand{\eprop}{\end{proposition}}
\newcommand{\bcor}{\begin{corollary}}
\newcommand{\ecor}{\end{corollary}}
\newcommand{\blem}{\begin{lemma}}
\newcommand{\elem}{\end{lemma}}

\def\RR{{\mathbb R}}
\def\A{{\cal A}}

\def\D{{\cal D}}

\def\M{{\cal M}}
\def\N{{\cal N}}
\def\R{{\cal R}}

\def\H{{\cal H}}

\def\S{{\cal S}}

\def\f{{\varphi}}

\def\hPhi{{\widehat \Phi}}

\def\emptyset{\varnothing}

\def\PSL{{{\rm PSL}(2,\mathbb R)}}

\def\S2{S^{1(2)}}

\def\RR{\mathbb R}
\newtheorem{theorem}{Theorem}[section]
\newtheorem{lemma}[theorem]{Lemma}

\newtheorem{corollary}[theorem]{Corollary}

\newtheorem{proposition}[theorem]{Proposition}

\theoremstyle{definition} 

\theoremstyle{remark} \newtheorem{remark}[theorem]{Remark}

\newcommand{\ben}{\begin{equation}}
\newcommand{\een}{\end{equation}}

\def\emptyset{\varnothing}

\def\PSL{PSU(1,1)}

\def\CC{{\mathbb C}}

\def\bS{{\mathbb S}}

\def\SL2{{{\rm SL}(2,\R)}}

\def\PSL2{{{\rm PSL}(2,\Reali)}}

\def\U1{{{\rm V}(1)}}
\def\SU2{{{\rm SV}(2)}}

\def\SU{{{\rm SU}}}

\def\A{{\mathcal A}}

\def\D{{\mathcal D}}

\def\H{{\mathcal H}}

\def\M{{\mathcal M}}
\def\N{{\mathcal N}}

\def\P{{\mathcal P}}

\title{\Huge{Modular Time Evolution 
and the QNEC
}}

\author{ {\sc Stefan Hollands}\\
Institut für Theoretische Physik, Universität Leipzig \\
Brüderstrasse 16, 04103 Leipzig, Germany\\
Max Planck Institute for Mathematics in Sciences (MiS)\\
Inselstra{\ss}e 22, 04103 Leipzig, Germany}

\date{}
\begin{document}

\maketitle

\begin{abstract}
We establish an inequality restricting the evolution of states in quantum field theory with respect to the modular flow of a wedge, $\Delta^{is}$, for large $|s|$. Our bound is related to the quantum null energy condition, QNEC. In one interpretation, it can be seen as providing a ``chaos-bound'' $\le 2\pi$ on the Lyapunov exponent with respect to Rindler time, $s$. Mathematically, our inequality is a statement about half-sided modular inclusions of von Neumann algebras.
\end{abstract}

\section{Introduction}
\label{intro}
A mathematical manifestation of the Unruh effect \cite{U76} is that the vacuum correlation functions of a relativistic
quantum field theory (QFT), when restricted to a wedge region of Minkowski space (Rindler space), look thermal in terms of the boost parameter (Rindler time) $s$,
\ben
x^0(s) = \cosh(s) x^0 + \sinh(s) x^1, \quad
x^1(s) = \sinh(s) x^0 + \cosh(s) x^1.
\een
The inverse temperature is given by\footnote{Our units are such that $k_B=c=\hbar=1$.} $\beta = 2\pi$ in this parameterization. A similar phenomenon occurs, e.g., for correlation functions of the Hartle-Hawking state when considered in the exterior of the Schwarzschild spacetime, or for the Bunch-Davies state, when considered in a wedge-like patch of deSitter spactime, see e.g., \cite{Hollands:2014eia}. 

This fact can be recognized by noting that the correlation functions of the Minkowski vacuum state are periodic in $s$ with period $2\pi$ after a Wick rotation, $s \to is$, or more precisely, that they satisfy the so-called KMS condition, see e.g., \cite{haag_2}. As is well-known, the KMS condition is closely related to Tomita-Takesaki theory of von Neumann algebras, and in this language, the Unruh effect is the statement that the modular flow
(see section \ref{modthbas}), $\Delta^{is}_\Omega$, of the vacuum state $\Omega$ with respect to the observable algebra of a wedge is equal to $\Delta^{is}_\Omega = e^{-2\pi is K}$, where $K$ is the generator of boosts in the $(x^0,x^1)$ plane. In this formulation, the Unruh effect is a consequence of the Bisognano-Wichmann theorem, see e.g., \cite{haag_2}. In this note, we study more closely the behavior of a time evolved state,
\ben
\label{Ps}
\Phi_s := \Delta_\Omega^{is} \Phi
\een
for asymptotically large Rindler time $s \to \infty$ or $s \to -\infty$, with respect to observables in the right wedge. 

Although not necessary for the technical arguments in this note, we now restrict attention to the toy model of chiral conformal QFTs (CFTs). These describe either the left- or right-moving chiral quantum fields in a two-dimensional CFT and already capture the essence of the Unruh/Bisognano-Wichmann setting. In the algebraic approach to QFT \cite{haag_2} that we adopt, a chiral CFT is a net $\{\A(I)\}$ of von Neumann algebras, acting on a common Hilbert space, $\H$, labeled by the proper, possibly half-infinite, intervals $I \subset \RR$. The real line hosting the intervals is thought of as parameterized by the null coordinate $x^+ = x^0 + x^1$ or $x^- = x^0-x^1$, in terms of which a boost corresponds to a dilation. 

A minimal\footnote{We do not require a stress tensor which enhances the symmetry to the Virasoro algebra.} set of axioms for a chiral net is:
\begin{itemize}
    \item Isotony: $I \subset J \Longrightarrow \A(I) \subset \A(J)$.
    \item Locality: $I \cap J = \emptyset \Longrightarrow [\A(I),\A(J)] = \{0\}$.
    \item Covariance: There exists a strongly continuous, projective, unitary representation $U(g), g \in M$ of the M\"obius group $M=\widetilde{SL_2}(\RR)/\{\pm 1\}$ on $\H$, acting by fractional linear transformations on $\RR$, such that $U(g) \A(I) U(g)^* = \A(gI)$ whenever $g$ is continuously connected to the identity by a path $t \mapsto g_t$ such that $g_t I$ does not contain the point at infinity. The generator of translations is positive. 
    \item Vacuum: There is a unique vector $\Omega \in \H$ (vacuum) such that $U(g)\Omega=\Omega$ for all $g \in M$.
    \item Additivity: $I=\cup_j I_j \Longrightarrow \A(I) = \vee_j \A(I_j),$ with $\vee$ the von Neumann closure of the union operation.
\end{itemize}
For a chiral CFT, the Bisognano-Wichmann theorem is the statement that the modular group $\Delta^{is}_\Omega$
with respect to $\A([0, \infty))$ corresponds to the representer, under $U$, of a dilation by $e^{2\pi s}$, or said differently $\Delta^{-is}_\Omega \A(I) \Delta^{is}_\Omega = \A(e^{2\pi s}I)$ for all intervals $I$. 

Let us consider the state functionals $\varphi_s$ (``reduced density matrices'') of the pure states $\Phi_s$ \eqref{Ps}, for the algebra $\A([0, b])$, where $b>0$ is arbitrary but fixed. The ergodic nature of the modular group implies that \cite{R74}
\ben
\label{ergodic1}
\lim_{s \to -\infty} \varphi_s(m) = \omega(m) \| 
\Phi\|^2 \qquad \forall m \in \A([0, b]),
\een
where $\omega$ is the state functional of the vacuum $\Omega$. Thus in this sense, at small distances any state looks like the vacuum. On the other hand, for $s \to + \infty$, the state functionals $\varphi_s$ may be expected to exhibit some sort of chaotic behavior. 

One purpose of this note is to provide mathematical bounds qualifying this statement in the sense of providing a sort of bound on chaos. Our bound, described in more detail in theorem \ref{thm1}, is phrased in terms of $S_{\rm meas}(\varphi | \! | \psi)_\M$, the so-called {\it measured relative entropy}. It is an information theoretic distance measure between two state functionals $\varphi, \psi$ on a von Neumann algebra $\M$ (``reduced density matrices'') similar to the relative entropy, see section \ref{modthbas}. 
We will show that 
\ben
\label{thm117}
\begin{split}
S_{\rm meas}(\varphi_s | \! | \varphi_s \circ T \circ R)_{\A([0,b])} \,
\le -bS' \cdot e^{2\pi s} + o\left( be^{2\pi s} \right).
\end{split}
\een
Here, $o(x)$ is a function going to zero faster than $x$ as $x \to 0$.\footnote{I.e., $\lim_{x \to 0} x^{-1}o(x) = 0$.} $S'$ is the derivative of the ordinary relative entropy 
\ben
S' = \frac{d}{da} S(\Phi  | \! | \Omega)_{\A([a,\infty))} \bigg|_{a=0},
\een
with respect to the observables localized in a half-ray. By monotonicity \cite{U77}, $S' \le 0$. $R$ is a reflection about the midpoint of $[0,b]$ and $T$
is a smoothed out version of a translation by a fixed amount of order one in Rindler time, $s$.

As we will discuss in section \ref{equilibsec}, a possible interpretation of these inequalities is the following. For $s \to -\infty$, the inequality says that, on each interval $[0,b]$ of fixed but possibly large length $b=\ell \gg 1$, the state $\Phi_{s}$ has an invariance under dilations and reflections up to an error of the order $\ell S' e^{2\pi s}$. This error obviously goes to zero as $s \to -\infty$ at an exponential rate. Since translations and reflections are symmetries fixing the vacuum, the bound can be thought of as a counterpart of ergodicity, \eqref{ergodic}.

The other limit $s \to \infty$ implies a kind of chaos bound. Here, we think of an interval $[0,b]$ as having very small length $b=\epsilon \ll 1$. 
We will see that our bound implies in particular that the expectation value of a sharply observable, $m$ in $[0,b]$, minus the same observable displaced by $\sim \epsilon$, i.e., $U(\epsilon)mU(\epsilon)^*$, does not grow faster than $\epsilon S' e^{2\pi s}$ under dynamical evolution with respect to dilations (parameter $s$), so long as $s$ is still small enough in order that $\epsilon e^{2\pi s} \lesssim 1$. 

It would be interesting to study relations between our chaos-type bound and a well-known one by \cite{Maldacena:2015waa}. Their bound on the Lyapunov exponent $\lambda_L \le 2\pi T$ \cite{Maldacena:2015waa} is similar to ours, $\lambda_L \le 2\pi$, in the Rindler space setting, but is phrased in terms of a rather different-looking observable, namely out-of-time-ordered thermal correlation functions. We leave a deeper study to a future investigation and only remark that the authors \cite{Maldacena:2015waa} have noted that their bound applies to the Rindler setting, though not directly to two-dimensional CFTs, indicating a difference to our approach. 

Instead we will describe in more detail in this note how our bound is connected, at least mathematically, to the so-called quantum null energy condition (QNEC) \cite{B1,B2,W1}, in a formulation involving relative entropies by Ceyhan and Faulkner \cite{CF18}. Actually, as \cite{CF18}, we will prove our bounds in a framework that only uses certain number of abstract features of chiral CFTs called {\it half-sided modular inclusions} \cite{Bor1, Bor2, Wies1, AZ05,Florig}. Our framework is thus more general and includes, e.g., situations in which the modular flow is not necessarily ergodic. 

Given that important aspects of our proof rely on the proofs of the QNEC \cite{CF18,HL25}, it would be interesting to understand better the physical relation between chaos bounds and the QNEC. At the mathematical level, 
our proof also uses an interpolation theorem \cite{SH3} of non-commutative $L_p$-spaces \cite{AM82}, and various variational characterizations of entropy besides the QNEC.


\section{Relative modular operators and entropy}
\label{modthbas} 

Consider a von Neumann algebra $\M$ acting on a Hilbert space $\H$. Suppose we have two cyclic and separating  vectors $\Om,\Phi\in\H$ for $\M$. Following \cite{Ar76,Ar77}, one can then define the {\it  relative Tomita's operators} \cite{Ar76,Ar77} on $\H$ as the closures of
\[
S_{\Om,\Phi} \equiv S_{\Om,\Phi;\M} : x\Phi \mapsto x^*\Om\, , \quad x\in \M\, ,
\]
\ben\label{SF}
S'_{\Om,\Phi} \equiv S_{\Om,\Phi;\M'}: x'\Phi \mapsto {x'}^*\Om\, , \quad x'\in \M'\, .
\een
The polar decompositions of these operators yield the {\it relative modular operators and conjugations}:
\[
S_{\Om,\Phi} = J_{\Om,\Phi}\Delta^{1/2}_{\Om,\Phi}\, ,\qquad S'_{\Om,\Phi} = J'_{\Om,\Phi}\Delta'^{1/2}_{\Om,\Phi}\, .
\]
Among the well-known formulas that we use are
\ben\label{JJ3}
J'_{\Om,\Phi} = J_{\Phi,\Om} = J^*_{\Om,\Phi}  \, , \qquad  {\Delta}'_{\Om,\Phi} = {\Delta}^{- 1}_{\Phi,\Om} = J_{\Om,\Phi}\Delta_{\Om,\Phi} J^*_{\Om,\Phi} \, ,
\een
as well as the covariance properties:
\ben
\label{covariance}
{\Delta}'_{v\Om,u\Phi} = u{\Delta}'_{\Om,\Phi}u^*, \quad
{\Delta}_{v'\Om,u'\Phi} =  u'{\Delta}_{\Om,\Phi}u'^*,
\een
for any isometries $u,v \in \M, u', v' \in \M$. 
The {\it modular flow} is the the 1-parameter group of automorphisms of $\M$ respectively $\M'$ given by, respectively
\ben
\label{mflow1}
\sigma_\Phi^t(m) = \Delta_\Phi^{it} m \Delta_\Phi^{-it}, \quad \sigma_\Phi'^t(m') = \Delta_\Phi^{\prime it} m' \Delta_\Phi^{\prime -it},
\een
where here and in the following we write
\ben
\Delta_\Phi := \Delta_{\Phi,\Phi}.
\een
One may also express the modular flows using the relative modular operators:
\ben
\label{mflow2}
\sigma_\Phi^t(m) = \Delta_{\Phi,\Omega}^{it} m \Delta_{\Phi,\Omega}^{-it}, \quad \sigma_\Phi'^t(m') = \Delta_{\Phi,\Omega}^{\prime it} m' \Delta_{\Phi,\Omega}^{\prime -it}.
\een
Let $\f = (\Phi, \cdot \Phi)$, $\om = (\Om, \cdot \Om)$ be the states on $\M$ associated with $\Phi,\Om$. Then we have the the following formulas for the {\it Connes-cocycle} \cite{C73}:
\ben\label{uDD}
u_s =(D\om : D\f)_s = \Delta^{is}_{\Om,\Phi}\Delta^{-is}_{\Phi} =
\Delta^{is}_{\Om}\Delta^{-is}_{\Phi,\Om}
\, .
\een
One may see that 
$u_s$ respectively $u_s'$ are unitary operators from $\M$ respectively $\M'$ for all $s \in \RR$.
It is also possible to construct unitary cocycles from the modular conjugations \cite[App. C]{AM82}, \cite[App. A]{CF18}:
\ben
v_{\Omega,\Phi} = J_{\Omega,\Phi}' J_\Phi' = J_{\Omega}' J_{\Omega,\Phi}' \in \M.
\een
With a certain amount of analysis, see \cite{Ar77}, \cite[App. C]{AM82}, \cite[App. A]{CF18}, these constructions and equalities can be 
generalized to situations where $\Phi$ is not cyclic and/or not separating. The corresponding modifications involve the so-called support projections $s(\Phi) \in \M, s'(\Phi) \in \M'$. Here, for example,
$s(\Phi)$ is the orthogonal projection onto the closure of the subspace $\M'\Phi$. We will avoid such considerations in this paper by generally assuming that $\Phi$ is cyclic and separating except in cases where this is explicitly mentioned. $\Omega$ is always assumed to be cyclic and separating. 

Araki's {\it relative entropy} is defined by \cite{Ar76,Ar77}
\ben
\label{Sdef}
S(\Phi|\!| \Om)_\M =  - (\Phi , \log \Delta_{\Om,\Phi}\Phi). 
\een
In view of the first covariance property \eqref{covariance}, one has
\ben
\label{Sdefcov}
S(u'\Phi|\!| v'\Om)_\M = S(\Phi|\!| \Om)_\M, \quad 
S(u\Phi|\!| v\Om)_{\M'} = S(\Phi|\!| \Om)_{\M'},
\een
for any isometries $u,v \in \M, u',v' \in \M'$. Therefore, $S(\Phi|\!| \Om)_\M \equiv S(\varphi|\!| \omega)_\M$, meaning that the relative entropy only depends on the 
functionals $\f = (\Phi, \cdot \Phi)$, $\om = (\Om, \cdot \Om)$, the states on $\M$ associated with $\Phi,\Om$. A similar statement applies to the relative entropy with respect to the commutant, see the second formula in \eqref{Sdefcov}. 

In the case of a finite-dimensional type $I$ von Neumann factor $\M$, $\varphi, \omega$ can be identified with self-adjoint operators having strictly positive eigenvalues, such that
$\varphi(m) = {\rm Tr}(\varphi m)$ for all $m\in \M$, and similarly for $\omega$. Under this identification, we have
\ben
S(\varphi|\!| \omega)_\M = {\rm Tr}(\varphi \log \varphi - \varphi \log \omega).
\een
There exist other equivalent characterizations of the relative entropy. One such characterization \cite[Thm. 9]{hiai} that we will use is 
\ben
\label{eq:var}
S(\Phi |\!| \Psi)_\M = \sup\{ (\Phi, x\Phi) - \log \| \Psi^x \|^2 \}, 
\een
with supremum over all self-adjoint elements $x$ of $\M$. Here, $\Psi^x \in \H$ is defined for any self-adjoint element $x \in \M$, and cyclic and separating vector $\Psi \in \H$ by the strongly convergent Araki perturbation series
\cite{Araki5}, 
\ben \label{psih}
\Psi^x = \sum_{n=0}^\infty \int_0^{1/2} ds_1 \int_0^{s_1} ds_2 \dots \int_0^{s_{n-1}} ds_n \, \Delta_\Psi^{s_n}x\Delta^{s_{n-1}-s_n}_\Psi x 
\cdots \Delta^{s_{1}-s_2}_\Psi x \Psi.
\een
The {\it measured relative entropy}, see e.g., \cite{OP}, is a commutative cousin of the Araki relative entropy. It may defined as 
\ben
\label{comvar}
S_{\rm meas}(\Phi | \! | \Psi)_\M = \sup \{ S(\Phi | \! | \Psi)_{\mathcal C} \mid {\mathcal C} \subset \M \},
\een
where $\mathcal C$ ranges over all {\it commutative} von Neumann subalgebras of $\M$. Like ordinary relative entropy, $S_{\rm meas}$ depends only on the expectation functionals $\varphi = (\Phi, \, . \, \Phi), \psi = (\Psi, \, . \, \Psi)$ on $\M$ induced by the vectors $\Phi,\Psi$. 
The terminology for $S_{\rm meas}$ is more apparent from the following equivalent \cite{OP} characterization
\ben
S_{\rm meas}(\Phi | \! | \Psi)_\M
= \sup \Bigg\{
\sum_i p_i \log \frac{p_i}{q_i} \, 
\bigg|
\, 
p_i=\varphi(e_i), q_i = \psi(e_i)
\Bigg\},
\een
were the supremum is over all finite sets $\{e_i\}$ of orthogonal projections $e_i \in \M$ such that
$e_i e_j = \delta_{ij}e_j, \sum_i e_i = 1$. We can think of these as eigenprojection of some hermitian observable from $\M$ whose corresponding eigenvalue is measured with probability $p_i=\varphi(e_i)$ respectively $q_i=\psi(e_i)$ in the states $\varphi$ respectively $\psi$.

As an application of the variational principle \eqref{comvar}, 
consider the commutative von Neumann subalgebra of $\M$
generated by the projection $e$ onto the positive support of 
the functional $\varphi-\psi$ and $1$. This yields \cite{Ar76} 
\ben
\label{Arakibound}
S_{\rm meas}(\Phi | \! | \Psi)_\M \ge \frac{1}{2}\| \psi-\varphi\|^2.
\een
The norm on the right side is that of a linear functional, $\rho: \M \to \CC$, defined as 
\ben
\|\rho\| = \sup\{|\rho(m)| \mid m \in \M, \|m\| = 1\}.
\een
In the case of a type $I$ von Neumann factor $\M$, a bounded linear functional $\rho: \M \to \CC$ can be identified with a trace-class operator such that
$\rho(m) = {\rm Tr}(\rho m)$ for all $m\in \M$. Under this identification, $\| \rho \|$ is the trace norm ${\rm Tr} \sqrt{\rho\rho^*}$.

The measured entropy may also be characterized by
\cite[Prop. 7.13]{OP}
\ben
\label{SMdef}
S_{\rm meas}(\Phi | \! | \Psi)_\M = \sup \{ (\Phi,x\Phi) - \log \| e^{x/2} \Psi \|^2 \}, 
\een
where the supremum is again over all self-adjoint elements $x$ of $\M$, compare \eqref{eq:var}.

\section{Half-sided modular inclusions}
\label{hsm}

Consider a proper inclusion $\N \subset \M$ of von Neumann algebras acting on the same Hilbert space $\H$, and assume that the following two conditions are satisfied.
\begin{enumerate}
\item There exists a unit vector $\Omega \in \H$ which is cyclic and separating for both $\M$ and $\N$. 
\item For any $t \ge 0$ it holds that $\Delta^{-it}_\Omega \N \Delta^{it}_\Omega \subset \N$, where $\Delta_\Omega$ is the modular operator for $\M$. 
\end{enumerate}
Then one calls the structure $(\N \subset \M, \Omega)$ a {\it half-sided modular inclusion}, with respect to $\Omega$. 

The point is that, for a half-sided modular inclusion, one has {\it two} modular operators associated with the same state $\Omega$ but different algebras, $\N$ and $\M$, satisfying the compatibility condition 2. It turns out that this has interesting consequences for the family of unitary operators 
\ben\label{Wiesbrock}
U(1-e^{-2\pi t}) = \Delta^{it}_{\Omega; \N }\Delta^{-it}_{\Omega;\M}. 
\een
The following structural result \cite{Wies1,AZ05,Florig} is called {\it Wiesbrock's theorem} for half-sided modular inclusions $(\N\subset\M, \Omega)$:
 
 \begin{theorem}
 \label{thm1}
There is a family of unitary operators $U(a), a \in \RR$ which is given by \eqref{Wiesbrock} for $a \le 1$, and which are realizing the situation described by Borchers' theorem \cite{Bor1,Bor2}:
 \begin{enumerate}
 \item $\RR \owns a \mapsto U(a)=e^{iaP}$ is a strongly continuous 1-parameter group of unitary operators on $\H$
 with positive generator $P$, meaning ${\rm spec} P  \subset [0,\infty)$.
 \item We have $U(a) \Omega = \Omega$ for all $a \in \RR$.
 \item We have $\M(a) := U(a) \M U(a)^* \subset \M$ when $a \ge 0$.
 \item We have $\M(1) = \N$.
 \item $\Omega$ is cyclic and separating for each $\M(a)$ and therefore for each $\M(a)'$.  If $a>0$, then $(\M(a) \subset \M,\Omega)$ is a half-sided modular inclusion. More generally, $(\M(b) \subset \M(a),\Omega)$ for $b > a$ are half-sided modular inclusions.
 \item We have $\Delta^{-it}_{\Omega} U(a) \Delta^{it}_{\Omega} = U(e^{2\pi t} a)$ for all $t,a \in \RR$, implying that $\Delta^{-it}_{\Omega} P \Delta^{it}_{\Omega} = e^{2\pi t} P$ for all $t \in \RR$, on the domain $\D(P)$ of $P$ given by Stone's theorem. In particular, $\D(P)$ is invariant under $\Delta^{it}_{\Omega}$.
These relations imply $\Delta_\Omega^{-it} \M(a) \Delta_\Omega^{it} =\M(e^{2\pi t}a)$.
 \item We have $\M(-a)'=U(-a) \M' U(-a)^* \subset \M'$ for $a \ge 0$ and $(\M(-a)' \subset \M',\Omega)$ is a half-sided modular inclusion when $a>0$. $J_{\Omega} U(a) J_{\Omega} = U(-a)$ for all $a \in \RR$.
 \end{enumerate}
 \end{theorem}

Chiral CFTs, see section \ref{intro}, give examples of half-sided modular inclusions, taking $\M=\A([0,\infty))$, $\N=\A([1,\infty))$ and $\Omega =$ the vacuum. Indeed, by the Reeh-Schlieder theorem (see e.g., \cite{haag_2}), $\Omega$ is cyclic for $\A(I)$ if $I$
has non-empty interior, so condition 1 for half-sided modular inclusions holds. By the Bisognano-Wichmann theorem, $\Delta_{\Omega}^{it}$ corresponds to a dilation by the factor $e^{2\pi t}$, and this factor is $\ge 1$ for $t \ge 0$. Thus, condition 2 holds, too.\footnote{Note that we also have Haag duality $\A([0,\infty))'=
\A((-\infty, 0])$.} By Borchers' and Wiesbrock's theorems, $U(a)=e^{iaP}$ corresponds to a translation by $-a$. In formulas, we have 
\ben
U(a) \A(I) U(a)^*=\A(I+a), 
\quad
\Delta_\Omega^{-it} \A(I) \Delta_\Omega^{it} =\A(e^{2\pi t}I)
\een
for any interval $I$.
In particular, we have $\M(a) = \A([a,\infty))$. 

Given a half-sided modular inclusion, arising form a chiral CFT or not, and a state $\Phi$, one may consider the relative entropies
with respect to $\M(a)$ respectively $\M(a)'$ as  functions of $a$:
\ben
\label{Sadef}
S(a):=S(\Phi|\!| \Om)_{\M(a)}, \quad
\bar S(a):=S(\Phi|\!| \Om)_{\M(a)'}.
\een
These entropies have a number of non-trivial properties. First of all, by monotonicity of the relative entropy \cite{U77}, we have
\ben
S(b) \le S(a) , \quad \bar S(b) \ge \bar S(a) \qquad \text{if $a \le b$.}
\een
Furthermore, by \cite[Lem. 1]{CF18} or \cite[Prop. 3.2]{HL25}, if $S(a), \bar S(b)<\infty$ and $b\ge a$, 
we have
\ben\label{dS1}
S(a) - S(b)   =   {\bar S}(a) - {\bar S}(b) + 2\pi(b-a) (\Phi, P\Phi) \, ,
\een
for every vector state $\Phi\in \D(P)$, the domain of $P$, which is called a {\it sum rule}. 

Combining the sum rule and monotonicity of the relative entropy, it follows that $S, \bar S$ are absolutely continuous on any interval  on which both $S, \bar S$ are finite, and the derivatives $\partial S, \partial \bar S$ exist almost everywhere on any such interval \cite{HL25b}. 

Next, suppose that the derivative $\partial S(a)$ exists for some $a$ and that $\Phi \in \D(P) \cap \D(\log \Delta'_a) \cap \D(\log \Delta_a)$, so, in particular,
$S(a), \bar S(a), (\Phi, P \Phi)<\infty$. Then \cite[Prop. 4.5]{HL25b}
\ben
\label{prop2s1}
\partial S(a) = i\left(\Phi, [P, \log \Delta_a'] \Phi \right),
\een
and the analogous formula also holds for $\partial \bar S(a)$,  
\ben
\label{prop2s2}
\partial \bar S(a) = i\left(\Phi, [P, \log \Delta_a] \Phi \right).
\een
Here we set
\ben
\Delta_a := \Delta_{\Phi,\Omega;\M(a)}, 
\quad 
\Delta_a' := \Delta_{\Phi,\Omega;\M(a)'}.
\een
Another nontrivial property of $S(a),\bar S(a)$, implying the QNEC, is the {\it ant formula}
\cite{W1}, in the formulation by \cite{CF18} using half-sided modular inclusions:
Suppose that for some $a \in \RR$, $\partial S(a)$ 
exists, 
$\Phi \in \D(P) \cap \D(\log \Delta'_a)$, 
and $u_{s}' \Phi \in \D(P)$ for $s > s_0$ and some $s_0$, where $u_s' = (D\omega' : D\varphi')_s$ is the Connes-cocycle associated with $\M(a)'$ and $\Phi,\Omega$. Then 
\ben
\label{antleft}
-\partial S(a) = 2\pi \inf_{u' \in \M(a)'} (u'\Phi, P u'\Phi) =
2\pi \lim_{s \to \infty}
(u_s'\Phi, P u_s'\Phi)
\een
Likewise, suppose that for some $a \in \RR$, $\partial \bar S(a)$ 
exists, 
$\Phi \in \D(P) \cap \D(\log \Delta_a)$,
and $u_{s} \Phi \in \D(P)$ for $s < -s_0$ and some $s_0$, where $u_s = (D\omega : D\varphi)_s$ is the Connes-cocycle associated with $\M(a)$ and $\Phi,\Omega$. Then
we have 
\ben
\label{antright}
\partial \bar S(a) = 2\pi \inf_{u \in \M(a)} (u\Phi, P u\Phi) = 2\pi \lim_{s \to -\infty}
(u_s\Phi, P u_s\Phi).
\een
The infima in \eqref{antleft}, \eqref{antright} are over isometries $u,u'$ in the respective algebras. See \cite[Thm. 5.1, Cor. 5.5]{HL25b} and \cite{CF18} for different proofs of the ant formulas \eqref{antright}, \eqref{antleft}.

\section{A bound related to the QNEC}
\label{equilibsec}

The ant formula \eqref{antright}
is saturated by the sequence of {\it flowed states} 
$u_s \Phi, s \to -\infty$, where $u_s = (D\omega:D\varphi)_s$ is the Connes-cocycle 
for the von Neumann algebra $\M(a)$. Consider the corresponding flowed state functional on $\M(a)$, given by
\ben
\label{varphisdef}
\varphi_s(m) := (u_s \Phi, m u_s \Phi) = 
(\Delta_{\Omega,a}^{is} \Phi, m \Delta_{\Omega,a}^{is} \Phi), \qquad m \in \M(a),
\een
where the second equality uses \eqref{uDD} and
\eqref{mflow2}. Recall that $\Delta_{\Omega,a}$ is the modular operator for the von Neumann algebra $\M(a)$ with respect to the vector $\Omega$ appearing in the definition of half-sided modular inclusion. We similarly define $\varphi^{\prime}_s = (u_{s}' \Phi, \ . \ u_{s}' \Phi)$, considered as a state functional on $\M(a)'$. 

Our aim is to investigate $\varphi_s$, and 
$\varphi^{\prime}_s$. To state our result, we introduce the linear operator $T: \M(a) \to \M(a)$
\ben
\label{Tdef}
T(m) := \int_\RR
 \Delta_{\Omega,a}^{-it} m \Delta_{\Omega,a}^{it} \,  \beta_{0}(t) \, d t, \quad m \in \M(a),
\een
which defines a unital, completely positive map, i.e., a {\it quantum channel}, see e.g., \cite{OP} for further discussion of this notion.
Here, 
\ben
\label{beta0def}
\beta_0(t) =  \frac{\pi}{2}
\frac{1}{\cosh (2\pi t) +1},
\een
which is a probability density on $\RR$. In a chiral CFT, $T$ implements a weighted average, by the probability density $\beta_0(t)$, over dilations by the factor $e^{2\pi t}$ around the point $a$. 

We also consider 
\ben
\M(a,b) = \M(a) \cap \M(b)' \qquad b>a.
\een
In a chiral CFT, this is a local algebra, i.e., we have $\M(a,b)=\A([a,b])$, though for a general half-sided modular inclusion, these algebras may be trivial, or they may be non-trivial only for a sufficiently large $|a-b|$. The linear operator $R: \M(a,b) \to \M(a,b)$
\ben
\label{Rdef}
R(m) := J_{\Omega, \frac{a+b}{2}} m^* J_{\Omega, \frac{a+b}{2}}, \quad m \in \M(a,b),
\een
defines a quantum channel. In a chiral CFT, it corresponds to a combined charge conjugation and reflection about the midpoint, $(a+b)/2$, of the interval $(a,b)$ by the Bisognano-Wichmann theorem. Here, $J_{\Omega,c}$ is the modular conjugation corresponding to $\Omega$ with respect to $\M(c)$ for some $c \in \RR$.

\begin{theorem}
\label{mainthm}
Suppose that for some $a \in \RR$, $\partial S(a)$ exists, $\Phi \in \D(P) \cap \D(\log \Delta_a)$, 
that $u_{s} \Phi \in \D(P)$ for all $s \in \RR$,
and that $c\varphi \le \omega \le c^{-1} \varphi$ for some $c>0$. Then the flowed state $\varphi_s$
\eqref{varphisdef} satisfies
\ben
\label{thm111}
S_{\rm meas}(\varphi_s | \! | \varphi_s \circ T \circ R)_{\M(a,b)}
\le (a-b)\partial S(a) \cdot e^{2\pi s} + o\left[ (b-a)e^{2\pi s} \right]
\een
for $a,b,s \in \RR$ such that $(b-a)e^{2\pi s} \to 0+$. Likewise, if $\partial \bar S(a)$ exists, $\Phi \in \D(P) \cap \D(\log \Delta_a')$, 
$u_{s}' \Phi \in \D(P)$ for all $s \in \RR$,
and $c\varphi' \le \omega' \le c^{-1} \varphi'$ for some $c>0$, then
\ben
\label{thm112}
S_{\rm meas}(\varphi^{\prime}_s | \! | \varphi^{\prime}_s \circ T \circ R)_{\M(b,a)}
\le (a-b) \partial \bar S(a) \cdot e^{-2\pi s} + o\left[ (a-b)e^{-2\pi s} \right]
\een
for $a,b,s \in \RR$ such that $(a-b)e^{-2\pi s} \to 0+$. In \eqref{thm111}, \eqref{thm112}, $o(x)$ denotes a function such that $\lim_{x \to 0} o(x)/x = 0$.
\end{theorem}

\begin{remark}
We refer to the appendix of \cite{HL25} for the construction of a wide class of states satisfying the assumptions of theorem \ref{mainthm}.
\end{remark}

\noindent
{\bf Discussion of theorem \ref{mainthm}.} First we combine \eqref{thm112} with the lower bound \eqref{Arakibound} on the measured relative entropy, considering $a,b$ such that $a-b=\ell$, where $0 < \ell$ is arbitrarily large but fixed. We consider $s>0$, chosen so large that 
\ben
\ell e^{-2\pi s} \lesssim 1,
\een
in order that the $o\left(\ell e^{-2\pi s} \right)$
term in 
\ben
\label{thm113}
\Big| \varphi_s'(m') - \varphi_s'(TR(m')) \Big|^2
\le  -2\ell \partial S(a)  e^{-2\pi s} + o\left(\ell e^{-2\pi s} \right),
\een
is negligible. Here, $m' \in \M(b) \cap \M(a)' = \M(b,a)$ with norm $\| m'\| =1$. In particular, let us consider the limit as $s \to \infty$. Then we learn that the norm of the functional $\varphi_s' - \varphi_s' \circ T \circ R$ on $\M(b,a)$ tends to zero when $s \to \infty$. 
Thus, for large $s$, $\varphi_s'$ is approximately invariant under the combined action of $R$, a reflection about the midpoint of $(a,b)$ \eqref{Rdef}, and $T$, the weighted averaging against the modular flow \eqref{Tdef}. The exact invariances under $R,T$ for any fixed $a,b$ characterize the vacuum state functional $\omega = (\Omega, \, . \, \Omega)$ in a chiral CFT. 

To interpret this, we notice that $\Phi_s' = \Delta_{\Omega, a}^{\prime is} \Phi$ implementing $\varphi_s'$ is the evolved state with respect to the modular dynamics. In a chiral CFT, this corresponds to a dilation by $e^{2\pi s}$ around $a$ in view of the Bisognano-Wichmann theorem. Thus, for $s \to \infty$, the excitations in the state $\Phi$ relative to the vacuum get dilated away to infinity from the viewpoint of observables $m' \in \M(b,a)$.  

As is well-known \cite[Thm. 3]{L79}, \cite[Thms. A, B]{MV24}, there is a strong connection between the existence of ergodic group actions on von Neumann algebras and the type $III_1$ property in the Connes classification. This applies in  particular to the setting of chiral CFT, where the modular flow $m' \mapsto \Delta_{\Omega, a}^{\prime -is} m' \Delta_{\Omega, a}^{\prime is} \equiv \sigma_{\Omega, a}^{\prime t}(m')$ is ergodic, and the local algebras for finite intervals $(a,b)$ are of type $III_1$, giving that \cite[Cor. 2.5]{R74}
\ben
\label{ergodic}
\lim_{s \to \infty} \varphi_s'(m') = \omega'(m') \| 
\Phi\|^2 \qquad \forall m' \in \M(b,a).
\een
Here, $\Phi \in \H$ is a pure state with associated state functional $\varphi' = (\Phi, \, . \, \Phi)$, and 
$\Omega$ is the vacuum with associated state functional $\omega' = (\Omega, \, . \, \Omega)$ on $\M(a)'$. Relation \eqref{ergodic} is a more direct way of stating and showing that excitations in the state $\Phi$ relative to the vacuum get dilated away from the viewpoint of $\M(b,a)$. It implies in view of $\omega' = \omega' \circ T \circ R$ 
that 
\ben
\label{thm115}
\lim_{s \to \infty} \Big( \varphi_s'(m') - \varphi_s'(TR(m')) \Big)
= 0.
\een
This is of course consistent with \eqref{thm113}, but unlike \eqref{thm113} does not show exponential decay as $s \to \infty$. 

Note at any rate, that theorem \ref{mainthm} and hence \eqref{thm113} holds for general half-sided modular inclusions, including ones for which the modular flow is not ergodic.

\medskip

Next, we consider consequences of \eqref{thm112}
for $s \to \infty$. Let $a,b$ such that $b-a = \epsilon$, where $0<\epsilon \ll 1$. Again, we consider $s>0$ and large, but, due to the different sign of $s$ in the exponentials compared to \eqref{thm111}, now still small enough that 
\ben
\epsilon e^{2\pi s} \lesssim 1, 
\een
in order that the $o\left(\epsilon e^{2\pi s} \right)$ term in
\ben
\label{thm114}
\Big| \varphi_s(m) - \varphi_s(TR(m)) \Big|^2
\le  2\epsilon \partial \bar S(a)  e^{2\pi s} + o\left(\epsilon e^{2\pi s} \right)
\een
be negligible. Here, $m \in \M(a,b)$ with $\| m \| = 1$. 

In particular, let us choose $m \in \M(a, a+\epsilon/2)$, such that $R(m) \in \M(b-\epsilon/2,b)$, is a translate of $m$ by the amount $\epsilon/2$, i.e., $R(m) = U(\epsilon/2) m U(\epsilon/2)^*$. For any $\xi, s_0>0$, we can find a smoothed out version $\psi$ of the state $\varphi$, such that $\| \psi_s \circ T -\psi_s\| < \xi$ for all $s<s_0$, see lemma \ref{smoothedout}. Define
\ben
\delta_\epsilon m :=m-U(\epsilon/2) m U(\epsilon/2)^*,
\een
which is the difference between $m$ and a small displacement by $O(\epsilon)$ of $m$.
It follows that 
\ben
|\psi_s(\delta_\epsilon m)| \le 2\epsilon \partial \bar S(a) e^{2\pi s}
+ o\left(\epsilon e^{2\pi s} \right) + \xi.
\een
This expresses a chaos bound, i.e. the expectation value of $\delta_\epsilon m$ in the evolved state $\psi_s$, grows exponentially at most at rate $e^{2\pi s}$ up to a small correction $\xi$ specifying the smoothness of the initial state, $\psi$. In this sense, we can say that the Lyapunov exponent is $\le 2\pi$. 

\begin{lemma}
\label{smoothedout}
Let $\varphi$ be a state on $\M(a)$, and $\xi, s_0 >0$. Then we can find $n$ such that $\|\psi_s \circ T -\psi_s\| < \xi$ for all $s<s_0$, where $\psi = \varphi \circ T^n$. 
\end{lemma}

\begin{proof}
Using the Fourier transform $\widehat \beta_0(k)$ of $\beta_0(t)$ \eqref{beta0def},
\ben
\widehat \beta_0(k) = \int_\RR e^{ikt} \beta_0(t) dt = \frac{k/2}{\sinh(k/2)},  
\een
we can write $T^n$ \eqref{Tdef} as 
\ben
T^n(m) = \frac{1}{2\pi} \int_{\RR^2}
 \Delta_{\Omega,a}^{-it} m \Delta_{\Omega,a}^{it} \,  e^{-ikt} \left( \frac{k/2}{\sinh(k/2)} \right)^n \, dt dk, \quad m \in \M(a). 
\een
Thereby, we can write 
\ben
\begin{split}
&|\psi_s \circ T(m) -\psi_s(m)|\\
= \, &
\frac{1}{2\pi}
\left|
\int_{\RR^2}
\varphi(\sigma_{\Omega,a}^t(m)) \,  e^{-ik(t-s)}\left( 
 \frac{k/2}{\sinh(k/2)} -1
 \right) \left( \frac{k/2}{\sinh(k/2)} \right)^n \, dt dk
\right| .
\end{split}
\een
The double integral can be estimated in an elementary way by $\le C\| m\| \| \varphi \| (1+ n^{-1}s^2) n^{-1/2}$, which yields the statement of the lemma by choosing $n$ sufficiently large.
\end{proof}

\section{Proof of theorem \ref{mainthm}}

For simplicity of notation, we put $a=0$ in this proof. This is not a loss of generality, since $(\M(a-1) \subset \M(a),\Omega)$ is isomorphic as a half-sided modular inclusion to $(\N=\M(-1)\subset \M = \M(0),\Omega)$. Consequently,
we simply write $\Delta_\Omega$ for $\Delta_{\Omega,a}$, etc.

The proof is partly based on an interpolation theory for non-commutative 
$L_p$-spaces and is similar in this respect to arguments presented in \cite[Sec. 5.3]{SH3}, see also \cite{Rac,Sutter2} for related arguments. 
The weighted $L_p$-spaces relative to a fixed vector $\Phi \in \H$ that we use were introduced by Araki-Masuda \cite{AM82}. For the case $p \ge 2$ needed in this proof, 
the space $L_p(\M, \Phi)$ for a von Neumann algebra $\M$ is 
the linear subspace of $\H$ defined by
\ben
L_p(\M, \Phi) = \left\{ \xi \in \bigcap_{\eta \in \H} \D(\Delta_{\eta, \Phi}^{(1/2)-(1/p)}) \ \bigg| \ \|\xi\|_{p,\Phi} < \infty \right\}.
\een
The norm is 
\ben
\|\xi\|_{p,\Phi} = \sup_{\|\eta\|=1} \left\|\Delta_{\eta, \Phi}^{(1/2)-(1/p)}\xi \right\|.
\een
See also e.g., \cite{Berta2, Hiai2} for further discussions of non-commutative $L_p$-spaces. 

It is known \cite{AM82} that the $L_2$-norm is the ordinary norm on $\H$ induced by the scalar product, $\| \xi \|_{2, \Phi} = \|\xi\| = \sqrt{(\xi,\xi)}$, while the $L_{\infty}$-norm is characterized by 
$\|m\Phi\|_{\infty,\Phi} = \|m\|$, for any $m \in \M$, where $\|m\| = \sup_{\|\xi\|=1} \|m\xi\|$ is the operator norm of $m$.

In the following we use the notation $\bS_I$  for a strip, where $I \subset \RR$ is an interval:
\ben
\label{strip}
\bS_I = \{ z \in \CC \mid \Im z \in I\}.
\een
We require the following interpolation theorem \cite{SH3}:
\begin{lemma}
\label{lem:hirsch}
Let $0<\theta<1/2$, $n \in {\mathbb N}$.
If $G(z)$ is an $\H$-valued function holomorphic on the strip ${\mathbb S}_{(0,1/2)}$
that is bounded and weakly continuous on the closure ${\mathbb S}_{[0,1/2]}$ and
such that 
\ben\label{boundcond}
\left\| G(i/(2n))\right\|_{n, \Phi}, \sup_{t \in \RR} \left\| G(t) \right\|_{\infty, \Phi}, \sup_{t \in \RR} \left\| G(t+i/2) \right\|_{2, \Phi} < \infty,
\een  
then we have
\ben
\label{himp0}
 \ln \left\| G(i/(2n)) \right\|_{n, \Phi} \leq   \int_\RR
  \left(
  \beta_{1-1/n}(t) \log \left\| G(t) \right\|_{\infty, \Phi}^{1-1/n} +  \beta_{1/n}(t) 
  \log  \left\| G(t+i/2) \right\|_{2, \Phi}^{1/n} \right) d t ,
  \nonumber
\een
where
\begin{equation}
\beta_\theta(t) = \frac{\sin(\pi\theta)}{2\theta[\cosh(2\pi t ) + \cos(\pi \theta)] }.
\end{equation}
\end{lemma}
We will exploit the consequences of this lemma for the $\H$-valued function 
\ben
\label{gta1}
G(t) = e^{-ith} U(b) \Delta^{-it}_{\Phi,\Omega} U(-be^{-2\pi t}) \Delta^{it}_{\Phi,\Omega} \Phi,  
\een
defined for a $h=h^* \in \M$ that we will chose later, $b \ge 0$, and at first for $t \in \RR$.

\blem
$G(z)$ can be extended to an $\H$-valued function holomorphic on the strip ${\mathbb S}_{(0,1/2)}$
that is bounded and weakly continuous on the closure
${\mathbb S}_{[0,1/2]}$. 
\elem
\begin{proof}
The statement can be proved by a repeated application 
of the following lemma to the factors in \eqref{gta1}. 

\begin{lemma}\label{proveanalyticity}
Suppose $F(z)$ is an $\H$-valued function which is analytic on $\bS_{(0,1/2)}$
and bounded and weakly continuous on the closure $ \bS_{[0,1/2]}$. Let $A(z), z \in \bS_{(0,1/2)}$ be a family of operators 
with common dense domain $\D$ such that $A(z)\eta$ is holomorphic on $\bS_{(0,1/2)}$, and strongly continuous on 
${\mathbb S}_{[0,1/2]}$
for $\eta \in \D$, 
and such that $C_0:= \sup_{t \in \RR} \| A(t) F(t) \|, C_1 := \sup_{t \in \RR} \| A(t+i/2) F(t+i/2)\|$ are finite.
Then $A(z) F(z)$ is an analytic function of $\bS_{(0,1/2)}$, 
which is bounded and weakly continuous on the closure 
${\mathbb S}_{[0,1/2]}$.
\end{lemma}

As stated, the lemma is a modest generalization of \cite[Lem. 3]{SH3}, which in turn is a slight generalization of 
\cite[Lem. 2.1]{CecchiniPetz}. Since the differences in the proof are so minor, it is omitted.

To check the assumptions of lemma \ref{proveanalyticity} in each factor of \eqref{gta1}, we can use
 the well-known analyticity and continuity properties of relative modular operator $A(z):=\Delta^{-iz}_{\Phi,\Omega}$ when applied to vectors 
 in the domain $\D=\D(\Delta^{1/2}_{\Phi,\Omega})$ and of $A(z):=\Delta^{iz}_{\Phi,\Omega} = \Delta'^{-iz}_{\Omega,\Phi}$ when applied to vectors 
 in the domain $\D=\D(\Delta'^{1/2}_{\Omega,\Phi})$ for $z \in \bS_{(0,1/2)}$. We also use
the fact that $\Im (-be^{-2\pi z})\ge 0$ for $z \in  {\mathbb S}_{[0,1/2]}$, so that the operator $P \ge 0$ in $A(z):=U(-be^{-2\pi z}) = \exp [i(-be^{-2\pi z})P]$ 
provides an exponential damping that causes $A(z)$ to satisfy the analyticity and continuity assumptions of Lemma \ref{proveanalyticity}. 
\end{proof}

Having extended $G(z)$ so that it satisfies the assumptions of lemma \ref{lem:hirsch}, we first evaluate:
\ben
\label{Ginfn}
\left\| G(t) \right\|_{\infty, \Phi} = \| e^{-ith} g(t,b)\Phi \|_{\infty, \Phi} = \| e^{-ith} g(t,b) \| \le 1,
\een
where 
\ben
\label{gta}
g(t,b) = U(b) \Delta^{-it}_{\Phi,\Omega} U(-be^{-2\pi t}) \Delta^{it}_{\Phi,\Omega}. 
\een
In the second equality in \eqref{Ginfn}, we used that $g(t,b)$ is from $\M$, as one can see \cite{HL25b} by commuting this operator with $m' \in \M'$ and using the relations of half-sided modular inclusions. 

In order to find an appropriate expression for $\| G(i/(2n)) \|_{n, \Phi}$ valid for large $n$, we first make use of the 
relations for half-sided modular inclusions to obtain the alternative expressions
\ben
\begin{split}
\label{Grep}
G(t)=& \, e^{-ith}U(b) \Delta_{\Phi,\Omega}^{-it}\Delta_{\Omega}^{it}U(-b) \Delta_{\Omega}^{-it}\Delta_{\Phi,\Omega}^{it}\Phi\\
=& \, e^{-ith}U(b) u_{-t}^{-1} U(-b) u_{-t} \Phi\\
=& \, u'e^{-ith}U(b) \Delta_{\hPhi,\Omega}^{-it}\Delta_{\Omega}^{\prime -it}U(-b) \Delta_{\Omega}^{-it}\Delta_{\Omega,\hPhi}^{\prime -it}\hPhi.
\end{split}
\een
Here $\hPhi \in \P^\sharp_\Omega$ is the representer of $\Phi$ in the natural cone with respect to $\M,\Omega$, and 
$u' \in \M'$ is an isometry such that $u'\hPhi = \Phi$ \cite{AM82}. In the second line in \eqref{Grep}, $u_t = (D\omega:D\varphi)_t$
is the Connes-cocycle \eqref{uDD}.
By well-known results, see e.g., \cite[Lem. 5]{SH2}, the family $u_z$ can be extended to a holomorphic function of $z \in \bS_{(-1/2,1/2)}$, if (and only if) $c\varphi \le \omega \le c^{-1} \varphi$ for some $c>0$, as we are assuming. We therefore have a Taylor series
\ben
\label{kdef}
u_z = 1 + ikz + O(z^2) \in \M
\een
for sufficiently small $|z|$, where $\|O(z^2)\| \le C |z|^2$ and where $k=\Delta_\Omega - \Delta_{\Phi,\Omega}$ is a self-adjoint element of $\M$. Using the penultimate 
representation in \eqref{Grep}, we may thereby write 
\ben
\label{kdef1}
G(i/(2n)) = e^{h/(2n)}\left[1 + \frac{U(b) k U(-b) - k}{2n} +O\left(\frac{1}{n^2}\right) \right]\Phi =: a_n \Phi,
\een
where the last equality defines $a_n \in \M$. 
When taking the $L_n$-norm of this expression, we first use a lemma, which has exactly the same proof as \cite[Lem. 6]{SH3}.

\blem
Let $n \in 4 \mathbb{N}$, $m \in \M$. Then $\| m \Phi\|_{n,\Phi}^n = \| (m \Delta_\Phi^{2/n} m^*)^{n/4}\Phi\|^2$.
\elem

Applying this lemma, we have 
\ben
\| G(i/(2n)) \|^n_{n,\Phi} = \| a_n \Phi \|_{n,\Phi}^n = \left\| (a_n \Delta_\Phi^{2/n} a_n^*)^{n/4}\Phi\right\|^2.
\een
Furthermore, by the same proof as for \cite[Lem. 7]{SH3}, we have
\ben
\label{limnG}
\lim_{n \to \infty} \| G(i/(2n)) \|^n_{n,\Phi} = \lim_{n \to \infty} \left\| (a_n \Delta_\Phi^{2/n} a_n^*)^{n/4}\Phi\right \|^2 = \|\Phi^{h/2-k_b/2}\|^2,
\een
with $k_b := U(b) k U(-b) - k$, see \eqref{kdef}, \eqref{kdef1}.
Here, $\Phi^x \in \H$ is defined for any self-adjoint element $x \in \M$ by \eqref{psih}.
Combining \eqref{limnG}, \eqref{Ginfn} with lemma \ref{lem:hirsch} therefore gives
\ben
\label{himp}
\log \| \Phi^{h/2-k_b/2} \|^2 \leq   \int_\RR
 \beta_{0}(t) 
  \log  \left\| G(t+i/2) \right\|  \ d t .
\een
We still require an expression for $\left\| G(t+i/2) \right\|$. This can be obtained from the last expression in \eqref{Grep}, using the 
standard relations and analyticity properties of modular operators, as well as $J_{\hPhi,\Omega} = J_{\Omega,\hPhi} = J_\Omega$ 
for vectors $\hPhi$ in the natural cone \cite{AM82}. We find
\ben
\begin{split}
\| G(t+i/2) \| &= \left\| e^{h/2} U(b) \Delta_{\hPhi,\Omega}^{-it+1/2}\Delta_{\Omega}^{\prime -it+1/2}U(-b) \Delta_{\Omega}^{-it+1/2}\Delta_{\Omega,\hPhi}^{\prime -it+1/2}\hPhi 
\right\| \\
&= \left\| e^{h/2} U(b) \Delta_{\hPhi,\Omega}^{-it+1/2}\Delta_{\Omega}^{\prime -it+1/2}U(-b) J_\Omega S_\Omega \Delta_{\Omega}^{-it}\Delta_{\hPhi,\Omega}^{it}\Omega 
\right\| \\
&= \left\| e^{h/2} U(b) \Delta_{\hPhi,\Omega}^{-it+1/2}\Delta_{\Omega}^{\prime -it+1/2}U(-b) J_\Omega S_\Omega u_{-t} \Omega 
\right\| \\
&= \left\| e^{h/2} U(b) \Delta_{\hPhi,\Omega}^{-it+1/2}\Delta_{\Omega}^{\prime -it+1/2}U(-b) J_\Omega u_{-t}^* \Omega \right\| \\
&= \left\| J_\Omega u_{-t}^* J_\Omega e^{h/2} U(b) \Delta_{\hPhi,\Omega}^{-it+1/2}\Delta_{\Omega}^{\prime -it+1/2} U(-b) \Omega \right\| \\
&= \left\| e^{h/2} U(b) \Delta_{\hPhi,\Omega}^{-it+1/2} \Omega \right\| \\
&= \left\| e^{h/2} U(b) \Delta_{\Phi,\Omega}^{-it}  J_{\Phi,\Omega}^*  S_{\Phi,\Omega} \Omega \right\| \\
&= \left\| e^{h/2} U(b) J_{\Phi,\Omega}^*  \Delta_{\Omega,\Phi}^{-it}   \Phi \right\| \\
&= \left\| e^{h/2} U(b) J_{\Omega,\Phi}  \Delta_{\Omega,\Phi}^{-it}   \Delta_{\Phi}^{it}  \Phi \right\| \\
&= \left\| e^{h/2} U(b) J_{\Omega,\Phi}  u_{-t} \Phi \right\|
\end{split}
\een
Equation \eqref{himp} thereby gives us
\ben
\label{himp1}
\begin{split}
2\log \| \Phi^{h/2-k_b/2} \|^2 & \leq   \int_\RR
 \beta_{0}(t) 
  \log  \left\| e^{h/2} U(b) J_{\Omega,\Phi}  u_{-t} \Phi \right\|^2  \ d t \\
  & \leq   \log \int_\RR
 \beta_{0}(t) 
\left\| e^{h/2} U(b) J_{\Omega,\Phi}  u_{t} \Phi \right\|^2  \ d t
 \end{split}
\een
using Jensen's inequality and $\beta_0(t)=\beta_0(-t)$ in the second step.
This relation is the basis for the next proposition:
\bprop
\label{prop3}
Assume $b \ge 0$, and $c \varphi \le \omega \le c^{-1} \varphi$ for some $c>0$. Then
\ben
\Big(\Phi, \Big\{ -U(b) \log \Delta_{\Phi,\Omega} U(-b) + \log \Delta_{\Phi,\Omega} + 2\pi b P \Big\} \Phi\Big) \ge S_{\rm meas}(\varphi |\! | \varphi_b)_\M,
\een
which uses, setting $\alpha_b(m)=U(-b)mU(b)$,
\ben
\label{vphib}
\varphi_{b}(m) := \int_\RR
 \beta_{0}(t) 
\left(  J_{\Omega,\Phi} u_{t} \Phi, \alpha_{b}(m) J_{\Omega,\Phi}  u_{t} \Phi \right)   d t, \quad m \in \M.
\een
\eprop
\begin{proof}
We take minus of inequality \eqref{himp1} and then add $\varphi(h)$ to both sides. Using also the chain rule $(\Phi^x)^y=\Phi^{x+y}$ \cite[Thm. 12.10]{OP}, we obtain
\ben
\label{himp2}
2\Big[\varphi(h/2) -\log \| (\Phi^{-k_b/2})^{h/2} \|^2 \Big] \ge  \varphi(h) - \log \varphi_b(e^h) .
\een
We next take the supremum over all self-adjoint elements $h \in \M$, and use the variational formula 
\eqref{SMdef}, as well as the variational formula 
\eqref{eq:var}. This shows 
\ben
\label{himp21}
2S(\Phi |\!| \Phi^{-k_b/2})_\M \ge  S_{\rm meas}(\varphi |\! | \varphi_b)_\M .
\een
Finally, we use  $S(\Phi |\!| \Phi^{x})  = -(\Phi, x\Phi)$ \cite[Thm. 3.10]{Ar77}, so
\ben
\label{himp22}
(\Phi,k_b \Phi) \ge  S_{\rm meas}(\varphi |\! | \varphi_b)_\M .
\een
Recall our formulas for 
$k_b$, that is $k=\Delta_\Omega - \Delta_{\Phi,\Omega}$ and $k_b := U(b) k U(-b) - k$. These formulas, our domain assumptions, and the relations for half-sided modular inclusions
which imply that $U(b) \log \Delta_\Omega U(-b) = 2\pi bP$, give 
\ben
\begin{split}
(\Phi,k_b \Phi) &= (\Phi, \{ U(b) k U(-b) - k\} \Phi) \\
&= 
(\Phi, \{ U(b) (\log \Delta_{\Omega} - \log \Delta_{\Phi,\Omega}) U(-b) - (\log \Delta_{\Omega} - \log \Delta_{\Phi,\Omega}) \} \Phi)\\
&= (\Phi, \{ -U(b) \log \Delta_{\Phi,\Omega} U(-b) + \log \Delta_{\Phi,\Omega} + 2\pi b P \} \Phi),
\end{split}
\een
proving the proposition.
\end{proof}
When $m \in \M(0,b)$, the formula \eqref{vphib} for $\varphi_b(m)$  simplifies, since
\ben
\label{phibsimpl}
\begin{split}
&\left(  J_{\Omega,\Phi} u_{t} \Phi, \alpha_{b}(m) J_{\Omega,\Phi}  u_{t} \Phi \right)\\
=&\left(  J_\Omega v'_{\Omega,\Phi} \Delta^{it}_{\Omega,\Phi} \Phi, \alpha_{b}(m) J_{\Omega,\Phi} v'_{\Omega,\Phi} \Delta^{it}_{\Omega,\Phi} \Phi \right)\\
=&\left(   v'_{\Omega,\Phi} \Delta^{it}_{\Omega,\Phi} \Phi, J_\Omega \alpha_{b}(m^*) J_{\Omega} v'_{\Omega,\Phi} \Delta^{it}_{\Omega,\Phi} \Phi \right)\\
=&\left(   v'_{\Omega,\Phi} \Delta^{it}_{\Omega,\Phi} \Phi, J_\Omega U(b)^* m^* U(b) J_{\Omega} v'_{\Omega,\Phi} \Delta^{it}_{\Omega,\Phi} \Phi \right)\\
=&\left(   v'_{\Omega,\Phi} \Delta^{it}_{\Omega,\Phi} \Phi, J_{\Omega,b/2} m^* J_{\Omega,b/2} v'_{\Omega,\Phi} \Delta^{it}_{\Omega,\Phi} \Phi \right)\\
=&\left(\Delta^{it}_{\Omega,\Phi} \Phi, J_{\Omega,b/2} m^* J_{\Omega,b/2}\Delta^{it}_{\Omega,\Phi} \Phi \right)\\
=&\left( \Phi, \Delta^{-it}_{\Omega} J_{\Omega,b/2} m^* J_{\Omega,b/2}\Delta^{it}_{\Omega} \Phi \right).
\end{split}
\een
Here we used that, by the relations for half-sided modular inclusions, $U(b)J_\Omega = J_{\Omega,b/2}$ is the reflection about the midpoint. Since this reflection leaves 
$\M(0,b)$ invariant, we get $x:=J_{\Omega,b/2} m^* J_{\Omega,b/2} \in \M(0,b)$, which commutes 
with the isometry $v'_{\Omega,\Phi}:=J_\Omega J_{\Omega,\Phi} \in \M'$, and satisfies
$\Delta^{-it}_{\Omega,\Phi}x\Delta^{it}_{\Omega,\Phi}=
\Delta^{-it}_{\Omega}x\Delta^{it}_{\Omega}$
by relative modular theory. 

Using the definitions of $T,R$ in the statement of the theorem and \eqref{phibsimpl} we get $\varphi_b(m) = \varphi(T\circ R(m))$ for all $m\in \M(0,b)$. The monotonicity of the measured relative entropy and proposition \ref{prop3} therefore yield
\ben
\label{xxx1}
\Big(\Phi, \Big\{ -U(b) \log \Delta_{\Phi,\Omega} U(-b) + \log \Delta_{\Phi,\Omega} + 2\pi b P \Big\} \Phi\Big) \ge S_{\rm meas}(\varphi |\! | \varphi \circ T \circ R)_{\M(0,b)}.
\een
This inequality also holds for $\Phi$ replaced by $u_s\Phi$, because $c^{-1} \varphi_s \le \omega \le c\varphi_s$ by the invariance of $\omega$ under the modular flow, 
using the notation $\varphi_s = (u_s \Phi, \ . \ u_s\Phi)$ as in the statement of the theorem (for $a=0$).
The other assumptions of the theorem also still hold for $u_s\Phi$, see \cite[Lem. 5.3]{HL25b}.
Having made the replacement, we evaluate the left side using the identity obtained by applying $\partial_t \ . \ |_{t=0}$ to the following equation \cite[Proof of lem. 5.4]{HL25b}:
\ben
\label{usform5}
(u_s \Phi, U(b) \Delta^{-it}_{u_s\Phi,\Omega} U(-be^{-2\pi t}) \Delta^{it}_{u_s\Phi,\Omega} u_s\Phi) = (\Phi, U(be^{2\pi s}) \Delta^{-it}_{\Phi,\Omega} U(-be^{-2\pi(t-s)}) \Delta^{it}_{\Phi,\Omega} \Phi).
\een
The equation \eqref{xxx1} thereby yields
\ben
\begin{split}
\label{bexp}
&\Big(\Phi, \Big\{ -U(be^{2\pi s}) \log \Delta_{\Phi,\Omega} U(-be^{2\pi s}) + \log \Delta_{\Phi,\Omega} + 2\pi b e^{2\pi s} P \Big\} \Phi\Big) \\
\ge & \, S_{\rm meas}(\varphi_s |\! | \varphi_s \circ T \circ R)_{\M(0,b)}.
\end{split}
\een
The domain assumptions of the theorem \ref{mainthm} imply that derivative with respect to $x:=be^{2\pi s}$ of the expression on the left side of \eqref{bexp} exists, meaning that 
\ben
\begin{split}
\label{aux22}
& x\partial_x
\Big(\Phi, \Big\{ -U(x) \log \Delta_{\Phi,\Omega} U(-x) + \log \Delta_{\Phi,\Omega} + 2\pi x P \Big\} \Phi\Big)_{x=0} + o(x) \\
 & \hspace{2cm} \ge S_{\rm meas}(\varphi_s |\! | \varphi_s \circ T \circ R)_{\M(0,b)}.
\end{split}
\een
(Recall that $o(x)$ donotes a function such that $\lim_{x \to 0} o(x)/x = 0$.)
By \cite[prop. 4.5]{HL25}, the derivative term on the left side of \eqref{aux22} is equal to $2\pi (\Phi, P \Phi) - \partial \bar S(0)$. This proves 
\ben
\label{thm11}
S_{\rm meas}(\varphi_s | \! | \varphi_s \circ T \circ R)_{\M(0,b)}
\le \Big[ 2\pi (\Phi, P\Phi) - \partial \bar S(0) \Big] be^{2\pi s} + o\left( be^{2\pi s} \right)
\een
for $be^{2\pi s} \to 0+$.
Finally, we use the sum rule \eqref{dS1} for the term in square brackets. This gives the statement of the theorem \ref{mainthm}.

\section{Outlook}
\label{outlook}
We end this note speculating about the possible significance of the states $\Phi_s = u'_s \Phi$ saturating the ant bound \eqref{antleft} as $s \to \infty$. Since the ant bound implies the QNEC in the formulation by \cite{CF18}, $\partial^2 S(a) \ge 0$, these states would also give an explicit error term in the QNEC when written more explicitly in terms of the expectation value of the stress tensor and the entanglement entropy \cite{B1,W1}.

In \cite{B3}, an interesting interpretation of these states was proposed in the context of holography. We think it would also be interesting to understand them better directly in the context of semi-classical gravity. E.g., one could have in mind a state $\Phi$ solving the semi-classical Einstein equations, describing a geometry representing a spherically symmetric collapsing null shell made up of a null-fluid, whose interior is a (perturbation of) Minkowski spacetime and whose exterior is a (perturbation of) Schwarzschild spacetime. One might then expect that $\Phi_s$ approaches a state equal to the Minkowski vacuum inside the shell and a (perturbation of) Schwarzschild spacetime outside. Expection values of stress tensors spacetimes representing a null shell with Minkowskian interior and Schwarzschild exterior have recently been computed in \cite{Ori:2025zhe}. 

\medskip

\noindent
{\bf Acknowledgements.} I warmly thank Roberto Longo for discussions.

\noindent
{\bf Data availability.} Data sharing not applicable to this article as no data sets were generated or analysed during the current study.

\medskip

\noindent
{\bf Declarations}

\medskip

\noindent
{\bf Conflict of interest.} The author has no relevant financial or non-financial interests to disclose. 

\end{document}